\documentclass{article}

\PassOptionsToPackage{numbers,sort&compress}{natbib}
\ifdefined\MLFourPSReview
  \usepackage{neurips_2026}
\else
  \usepackage[preprint]{neurips_2026}
\fi
\usepackage[utf8]{inputenc}
\usepackage[T1]{fontenc}
\usepackage{microtype}

\usepackage{amsmath,amssymb,amsthm,mathtools,bm}
\newtheorem{theorem}{Theorem}[section]

\theoremstyle{definition}

\theoremstyle{remark}

\usepackage{graphicx}
\usepackage{booktabs}
\usepackage{subcaption}
\usepackage{algorithm}
\usepackage{algpseudocode}
\graphicspath{{figures/}}

\usepackage[hidelinks]{hyperref}
\usepackage{orcidlink}
\usepackage[capitalise,noabbrev,nameinlink]{cleveref}
\ifdefined\MLFourPSReview
  \hypersetup{pdfauthor={},pdfsubject={}}
  \makeatletter
  \renewcommand{\@noticestring}{Submitted to the 9th Workshop on Machine Learning and the Physical Sciences (ML4PS 2026). Do not distribute.}
  \makeatother
\else
  \hypersetup{pdfauthor={Rafael Coelho Lopes de Sa and Jay Sandesara}}
\fi

\hypersetup{pdftitle={Finite Asimov Sample Construction in Unbinned Neural Simulation-Based Inference}}

\title{Finite Asimov Sample Construction in\\ Unbinned Neural Simulation-Based Inference}

\author{%
  Rafael Coelho Lopes de Sa\,\orcidlink{0000-0001-5200-9195} \\
  Department of Physics \\
  University of Massachusetts Amherst \\
  MA, USA \\
  \texttt{rclsa@umass.edu}
  \And
  Jay Sandesara\,\orcidlink{0000-0002-6016-8011} \\
  Data Science Institute \\
  University of Wisconsin--Madison \\
  WI, USA \\
  \texttt{jsandesara@wisc.edu}
}

\begin{document}
\maketitle

\begin{abstract}
Expected sensitivity calculations in frequentist analysis using neural simulation-based inference often require large simulated samples to approximate an unbinned Asimov dataset. We construct a finite weighted reference sample, for analyses based on learned density ratios, for which the generating parameters globally maximize the weighted likelihood. A toy example with five-dimensional observable space based loosely on high-energy physics models shows that an Asimov dataset consisting of as few as 256 weighted reference events closely reproduce the expected test statistic scans obtained with two million simulated reference events. The resulting test statistic distributions are also shown to agree with independent simulator-based pseudo-experiments. We note that this agreement with the simulator depends on the accuracy of the learned ratios and must be validated explicitly.
\end{abstract}

\section{Introduction}

Neural ratio estimation (NRE) enables frequentist inference with an intractable likelihood by learning its ratio to some parameter-independent reference distribution,
\begin{equation}
  r_{\boldsymbol\psi}(\mathbf{x};\boldsymbol\theta)
  \simeq \frac{p(\mathbf{x};\boldsymbol\theta)}{p_{\rm ref}(\mathbf{x})}.
  \label{eq:nre}
\end{equation}
Here $\mathbf{x}\in\mathbb{R}^d$ denotes a complete and possibly correlated observation. A classifier $D_{\boldsymbol\psi}(\mathbf{x};\boldsymbol\theta)$ trained to distinguish target from reference samples with balanced class weights provides the ratio estimate through its odds, $r_{\boldsymbol\psi}=D_{\boldsymbol\psi}/(1-D_{\boldsymbol\psi})$~\cite{Cranmer:2015bka}. Writing $\boldsymbol\theta=(\mu,\boldsymbol\alpha)$ for a parameter of interest $\mu$ and nuisance parameters $\boldsymbol\alpha$, this ratio prediction can be used to compute the profile negative log-likelihood ratio statistic
\begin{equation}
\begin{gathered}
  t_\mu=-2\log\frac{p(\mathbf{x};\mu,\widehat{\widehat{\boldsymbol\alpha}}_\mu)}
  {p(\mathbf{x};\widehat\mu,\widehat{\boldsymbol\alpha})},\\
  (\widehat\mu,\widehat{\boldsymbol\alpha})=\arg\max_{\mu,\boldsymbol\alpha}p(\mathbf{x};\mu,\boldsymbol\alpha),
  \qquad \widehat{\widehat{\boldsymbol\alpha}}_\mu=\arg\max_{\boldsymbol\alpha}p(\mathbf{x};\mu,\boldsymbol\alpha).
\end{gathered}
\label{eq:profile}
\end{equation}
Its sampling distribution determines confidence intervals~\cite{Cowan:2010js}.  Replacing $p$ by $r_{\boldsymbol\psi}$ throughout Eq.~\eqref{eq:profile} gives the NRE approximation, since the parameter-independent reference cancels from this statistic and leaves the maximizers unchanged. This underlies the ATLAS neural simulation-based inference framework and its off-shell Higgs measurement~\cite{ATLAS:2024rpr,ATLAS:2024tgo}.

In high-energy physics (HEP), the usual specialization combines a Poisson event count, conditionally iid events $\{\mathbf{x}_i\}$, a mixture of processes, and auxiliary constraint terms to build an extended likelihood model:
\begin{equation}
  \mathcal L(\boldsymbol\theta)\propto
  e^{-\lambda(\boldsymbol\theta)}\prod_{i=1}^{n}\nu(\mathbf{x}_i;\boldsymbol\theta)
  f(\mathbf a;\boldsymbol\alpha),\qquad
  \nu=\sum_s\lambda_s p_s,\quad \lambda=\int\nu(\mathbf{x},\boldsymbol{\theta})\,d\mathbf{x}=\sum_s\lambda_s.
  \label{eq:hep}
\end{equation}
Here $s$ labels components such as signal and background, each with a normalized density $p_s$ and expected event count $\lambda_s$. Their combined intensity $\nu$ gives the expected number of events in an infinitesimal region $d\mathbf{x}$ as $\nu(\mathbf{x};\boldsymbol\theta)d\mathbf{x}$. The factor $f(\mathbf a;\boldsymbol\alpha)$ is the likelihood of independent auxiliary measurements $\mathbf a$ that constrain the nuisance parameters $\boldsymbol\alpha$.

An Asimov dataset represents the expected experiment with a fit that returns the generating parameters~\cite{Cowan:2010js}. In a binned analysis, each bin is filled with its expected count. ATLAS approximated the unbinned counterpart using a large weighted Monte Carlo (MC) sample~\cite{ATLAS:2024rpr}. Finite MC fluctuations displace its likelihood maximum, while repeated evaluations on millions of events make expected scans expensive. ATLAS reports likelihood maximizations requiring $\mathcal O(10-20)$ CPU hours and $\mathcal O(100-500)$~GB in this regime. We address the finite-sample displacement directly, allowing the integration sample to be chosen for scan accuracy rather than approximate score cancellation.

\section{Finite sample unbinned Asimov construction}
\label{sec:construction}

Let $r_{s,\boldsymbol\psi}(\mathbf{x};\boldsymbol\theta)$ be fixed, finite learned estimates of the ratios $p_s/p_{\rm ref}$ between a target and a reference density. We also allow signed component densities $p_s$ and signed integrated yields $\lambda_s$. Signed ratios can be learned with the alternative loss functionals~\cite{Drnevich:2024PARE} or with other methods that do not use neural networks. Choose reference points $\mathcal X_M=\{(\mathbf{x}_m,\omega_m)\}_{m=1}^M$, with $\omega_m>0$ and $\sum_m\omega_m=1$. Ordinary draws from $p_{\rm ref}$ have $\omega_m=1/M$. The reference must cover the target support so that $r_{s,\psi}$ are finite. Define, on these same points at every tested parameter value,
\begin{equation}
  Z_s(\boldsymbol\theta)=\sum_m\omega_m r_{s,\boldsymbol\psi}(\mathbf{x}_m;\boldsymbol\theta),
  \qquad
  h(\mathbf{x};\boldsymbol\theta)=\sum_s\lambda_s(\boldsymbol\theta)
  \frac{r_{s,\boldsymbol\psi}(\mathbf{x};\boldsymbol\theta)}{Z_s(\boldsymbol\theta)}.
  \label{eq:normalization}
\end{equation}
Assume each $Z_s$ is finite and nonzero. Zero-integral components require a different decomposition or normalization of the positive total intensity. At a single pre-determined generating point $\boldsymbol\theta_A = (\mu_A, \boldsymbol\alpha_A)$, assign Asimov weights equal to the finite intensity masses:
\begin{equation}
  v_m(\boldsymbol\theta)=\omega_m h(\mathbf{x}_m;\boldsymbol\theta),
  \qquad w_m^A=v_m(\boldsymbol\theta_A),
  \qquad \sum_m v_m(\boldsymbol\theta)=\lambda(\boldsymbol\theta).
  \label{eq:masses}
\end{equation}
Thus $\sum_m w_m^A=\lambda(\boldsymbol\theta_A)$ exactly, independently of $M$. With fixed auxiliary observations $\mathbf a_A$, the weighted log likelihood, up to parameter-independent terms, is
\begin{equation}
  \ell_{A,M}(\boldsymbol\theta)
  =-\lambda(\boldsymbol\theta)+\sum_m w_m^A\log h(\mathbf{x}_m;\boldsymbol\theta)
  +\log f(\mathbf a_A;\boldsymbol\alpha).
  \label{eq:asimov}
\end{equation}
Algorithm~\ref{alg:nre-asimov} summarizes the construction.

\begin{algorithm}[htbp]
\caption{Finite Asimov construction with NRE}
\label{alg:nre-asimov}
\begin{algorithmic}[1]
\Require Sampler for $p_{\rm ref}$; fixed functions $\{r_{s,\boldsymbol\psi},\lambda_s\}$; $\boldsymbol\theta_A = (\boldsymbol{\mu}_A, \boldsymbol\alpha_A)$; sample size $M$; constraint $f$.
\Ensure Weighted sample $\mathcal A_M$ and log-likelihood function $\ell_{A,M}$.
\State Draw $\mathbf{x}_m\sim p_{\rm ref}$, set $\omega_m=1/M$, and keep these points fixed.
\Function {$H$}{$\boldsymbol\theta$} \Comment{Called at every scan or profiling point}
  \State $Z_s(\boldsymbol{\theta})\gets\sum_m\omega_m r_{s,\boldsymbol\psi}(\mathbf{x}_m;\boldsymbol\theta)$ for every process $s$.
  \State \Return $\bigl\{h(\mathbf{x}_m;\boldsymbol\theta)=\sum_s\lambda_s(\boldsymbol\theta)r_{s,\boldsymbol\psi}(\mathbf{x}_m;\boldsymbol\theta)/Z_s(\boldsymbol{\theta})\bigr\}_{m=1}^M$.
\EndFunction
\State $\{h_m^A\}\gets H(\boldsymbol\theta_A)$; set $w_m^A=\omega_m h_m^A$ and $\mathcal A_M=\{(\mathbf{x}_m,w_m^A)\}$.
\State Verify $\sum_m w_m^A=\lambda(\boldsymbol\theta_A)$; choose $\mathbf a_A$ with $\boldsymbol\alpha_A\in\arg\max_{\boldsymbol\alpha}f(\mathbf a_A;\boldsymbol\alpha)$.
\State \Return $\mathcal A_M$ and $\ell_{A,M}(\boldsymbol\theta)=-\lambda(\boldsymbol\theta)+\sum_m w_m^A\log h(\mathbf{x}_m;\boldsymbol\theta)+\log f(\mathbf a_A;\boldsymbol\alpha)$.
\end{algorithmic}
\end{algorithm}

\begin{theorem}[Finite sample Asimov closure]
Suppose $0<v_m(\boldsymbol\theta)<\infty$ for every $m$ throughout the parameter domain. Choose $\mathbf a_A$ such that $f(\mathbf a_A;\boldsymbol\alpha)$ is globally maximized at $\boldsymbol\alpha_A$, with $0<f(\mathbf a_A;\boldsymbol\alpha_A)<\infty$. Then $\boldsymbol\theta_A$ is a global maximizer of $\ell_{A,M}$ for every $M$.
\end{theorem}
\begin{proof}
Write $v_m^A=v_m(\boldsymbol\theta_A)$ and $u_m=v_m(\boldsymbol\theta)/v_m^A>0$. Using Eq.~\eqref{eq:masses},
\begin{align}
  \ell_{A,M}(\boldsymbol\theta)-\ell_{A,M}(\boldsymbol\theta_A)
  &=-\lambda(\boldsymbol\theta)+\lambda(\boldsymbol\theta_A)
  +\sum_m v_m^A\log u_m
  +\log\frac{f(\mathbf a_A;\boldsymbol\alpha)}{f(\mathbf a_A;\boldsymbol\alpha_A)}
  \nonumber\\
  &=-\sum_m v_m^A\left[u_m-1-\log u_m\right]
  +\log\frac{f(\mathbf a_A;\boldsymbol\alpha)}{f(\mathbf a_A;\boldsymbol\alpha_A)}
  \leq0.\label{eq:proof}
\end{align}
Each $v_m^A>0$, each bracket is nonnegative, and the auxiliary term is nonpositive. Equality holds exactly when $u_m=1$ for every $m$ and $f(\mathbf a_A;\boldsymbol\alpha)=f(\mathbf a_A;\boldsymbol\alpha_A)$, as at $\boldsymbol\theta_A$.
\end{proof}

No differentiability is required; for differentiable models at an interior generating point, the score vanishes. The maximum may be nonunique. The condition on the auxiliary terms is satisfied for Gaussian constraints, but not for arbitrary constraint families. Strict positivity is required for the combined fitted intensity at the fixed reference points throughout the scan. Signed components and model misspecification preserve closure if total positivity and normalization hold, but agreement with the simulator requires independent validation.

The normalizers in Eq.~\eqref{eq:normalization} must include the complete parameter-dependent ratios and be recomputed on the same fixed $\mathcal X_M$ during profiling. A nominal normalizer suffices for parameters that affect only yields. Shape-dependent ratios generally require updating $Z_s(\boldsymbol\theta)$, freezing these normalizers need not preserve closure in shape directions. In HEP, score directions associated with signal-strength and normalization systematics only require nominal normalization, while the case for shape systematics will depend on the specific implementation.

Note that, while the test statistic $t_{\mu,A}=2[\ell_{A,M}(\boldsymbol\theta_A)-\max_{\boldsymbol\alpha}\ell_{A,M}(\mu,\boldsymbol\alpha)]$ is nonnegative, its curvature and values away from the minimum still require convergence with $M$. Independent reference samples can assess the remaining integration error, including fluctuations of the normalizers. This correction neither retrains the estimator nor removes local errors in its learned ratios. 

To quantify discovery significance at $\mu_A>0$, set $q_{0,A}=t_{0,A}$. Under the usual large-event yield, local Wald assumptions~\cite{Cowan:2010js} hold,
\begin{equation}
  \sigma_A\simeq\frac{\mu_A}{\sqrt{q_{0,A}}},\qquad
  Z_A\simeq\sqrt{q_{0,A}}.
  \label{eq:wald}
\end{equation}
These are asymptotic approximations, separate from the exact finite-$M$ theorem.

\section{Demonstration on a toy example}
\label{sec:demonstration}

We choose a toy HEP-inspired statistical model where a signal hypothesis is scaled linearly by signal strength parameter $\mu$ and background hypothesis is kept fixed. Each sampled event spans a five-dimensional observable space modeled using correlated Gaussian mixtures and Gaussian detector smearing, and undergoes a fixed classifier selection. The selected expected yields are approximately $\lambda_S=0.5\times 10^3$ and $\lambda_B=150\times 10^3$, with intensity $\nu=\mu\lambda_Sp_S+\lambda_Bp_B$, $\mu\geq0$, and no nuisance parameters. The reference is $p_{\rm ref}=(p_S+p_B)/2$ after selection. For each process, four binary classifiers with four 1024-unit Swish hidden layers estimate $p_s/p_{\rm ref}$; their ratios are averaged. Each of the signal, background, and reference training samples contains five million events. Independent validation samples contain $5\times10^5$ events each. 

\begin{figure}[htbp!]
  \centering
  \begin{subfigure}{0.49\textwidth}
    \includegraphics[width=\linewidth]{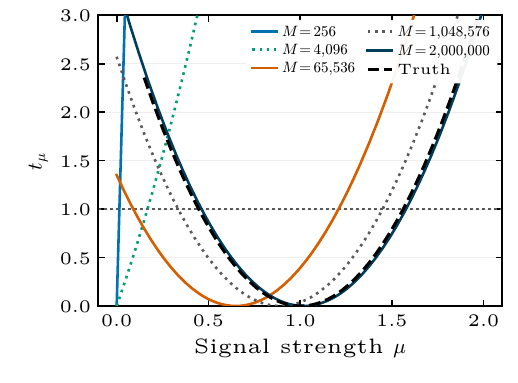}
    \caption{Conventional simulated samples}
  \end{subfigure}\hfill
  \begin{subfigure}{0.49\textwidth}
    \includegraphics[width=\linewidth]{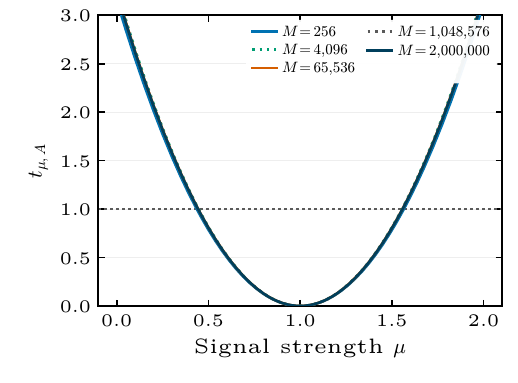}
    \caption{Normalized Asimov samples}
  \end{subfigure}
  \caption{\small Expected test statistic scans at $\mu_A=1$. (a) Conventional signal and background samples. The dashed black analytic curve uses five million selected events per process and is shifted horizontally to place its sampled minimum at $\mu=1$. (b) Reference samples with the normalization and weights of Eqs.~\eqref{eq:normalization}--\eqref{eq:masses}.}
  \label{fig:scans}
\end{figure}

At $\mu_A=1$, we compare conventional samples containing $M/2$ signal and $M/2$ background events, weighted by $2\lambda_S/M$ and $2\lambda_B/M$, with $M$ reference events treated by Eqs.~\eqref{eq:normalization}--\eqref{eq:asimov}. We use eight independent repetitions and nested sizes $M=256$ to $2\times10^6$. A separate five-million-event reference sample fixes the normalization of the likelihood used for conventional samples and all toy fits. Each smaller corrected sample instead defines its own finite model.

Figure~\ref{fig:scans} shows the first repetition. Already at $M=256$, the mean $q_{0,A}$ is $3.20$ with a between-repetition standard deviation of $0.11$, compared with $3.24$ for the large reference sample. The conventional $M=256$ scan shown has its minimum at the boundary; even at $M=2\times10^6$, the conventional fitted minima have standard deviation $0.09$ across repetitions. The gain is therefore in both exact location and much faster stabilization of expected sensitivity.

\begin{figure}[htbp!]
  \centering
  \begin{subfigure}{0.49\textwidth}
    \includegraphics[width=\linewidth]{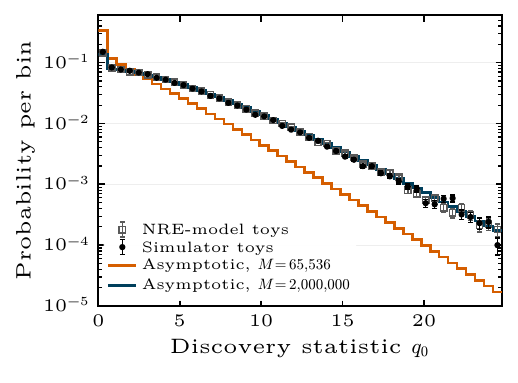}
    \caption{Conventional predictions}
  \end{subfigure}\hfill
  \begin{subfigure}{0.49\textwidth}
    \includegraphics[width=\linewidth]{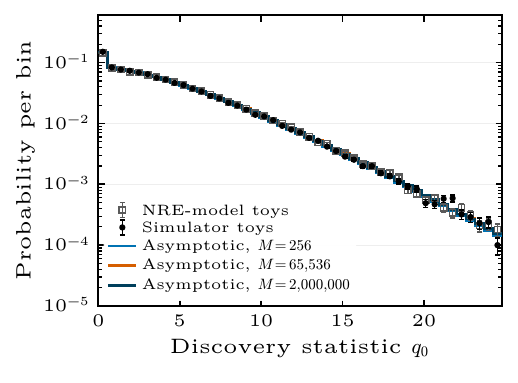}
    \caption{Normalized Asimov predictions}
  \end{subfigure}
  \caption{\small Discovery statistic at $\mu_A=1$ for $10^5$ NRE-model toys (gray open squares) and $10^5$ simulator-bank toys (black circles), both with Poisson errors, compared with the asymptotic approximation based on the width Eq.~\eqref{eq:wald} evaluated using (a) conventional and (b) corrected Asimov likelihood fits. Errors reflect toy counts conditional on the finite simulation banks. The two set of markers largely overlap, empirically demonstrating that the density ratios are well trained.}
  \label{fig:toys}
\end{figure}

We generate $10^5$ pseudo-experiments from each of two sources: the ratio-weighted large reference sample obtained by reweighing reference sample toys, and independent simulator banks of five million events per process. Figure~\ref{fig:toys} shows that small corrected samples already give useful predictions of the discovery distribution using the non-central $\chi^2$ asymptotic approximations, whereas conventional predictions are sensitive to MC fluctuations. The two toy sources agree closely here. This is empirical and the result of well-trained density ratios, not guaranteed by normalization. Residual differences, especially in sparsely populated tails, can arise from ratio error, finite simulation banks, and the asymptotic approximation. Simulator validation and, where needed, calibration remain necessary.

\section{Conclusion}

We showed that a finite sample Asimov dataset can be constructed by requiring a consistent ratio normalization. In the demonstration, a quadrature with a few hundred integration points, comparable in number to the bins of a large binned analysis, suffice for useful expected inference. This addresses a major computational burden, as demonstrated for example in the the ATLAS experiment measurement with an NSBI workflow, by reducing the size of the sample used for building expected test statistic scans. We note that it does not, however, reduce the simulation size needed to train and validate the learned likelihood, and the required integration size remains analysis dependent.
All source code and notebooks required to reproduce this demonstration are publicly available~\cite{NREAsimovDemo}.

\bibliographystyle{unsrtnat}
\bibliography{references}
\end{document}